\documentclass[aps,pra,twocolumn]{revtex4-1}

\usepackage{amsmath}
\usepackage{amsthm}

\newtheorem{theorem}{Theorem}[section]

\usepackage{bbm,graphicx,bm}   
\usepackage[caption=false]{subfig}

\makeatletter
\newcommand*{\defeq}{\mathrel{\rlap{%
                     \raisebox{0.3ex}{$\m@th\cdot$}}%
                     \raisebox{-0.3ex}{$\m@th\cdot$}}%
                     =}
\makeatother

\def\ii{{\rm i}}
\def\sx{\sigma^{\rm x}}
\def\sy{\sigma^{\rm y}}
\def\sz{\sigma^{\rm z}}
\def\tr#1{{\rm tr}(#1)}
\def\1{\mathbbm{1}}
\def\cL{{\cal L}}
\def\cLd{{\cal L}_{\rm dis}}
\def\cLdd{{\cal L}_{\rm dis}^{(\rm d)}}
\def\ket#1{{| #1 \rangle}}

\def\tit#1{{\em #1},}
\def\etal#1{#1}

\begin{document}

\title{Geometry of local quantum dissipation and fundamental limits to local cooling}

\author{Marko \v Znidari\v c}
\affiliation{Physics Department, Faculty of Mathematics and Physics, University of Ljubljana, Ljubljana, Slovenia}

\date{\today}

\begin{abstract}
We geometrically characterize one-qubit dissipators of a Lindblad type. An efficient parametrization in terms of $6$ linear parameters opens the way to various optimizations with local dissipation. As an example, we study maximal steady-state singlet fraction that can be achieved with an arbitrary local dissipation and two qubit Hamiltonian. We show that this singlet fraction has a discontinuity as one moves from unital to non-unital dissipators and demonstrate that the largest possible singlet fraction is $\approx 0.654$. This means that for systems with a sufficiently entangled ground state there is a fundamental quantum limit to the lowest attainable energy. With local dissipation one is unable to cool the system below some limiting non-zero temperature.
\end{abstract}

\pacs{03.65.Yz, 03.65.Aa, 03.67.Bg}



\maketitle

\section{Introduction}

Quantum information theory aims at finding protocols that outperform best classical ones~\cite{Nielsen}. This can be achieved either on a case-by-case basis, studying specific problems and finding efficient algorithms, or, one can try to characterize the power of a certain generic resource, like e.g. entanglement, local interaction etc.. The latter method can be very powerful, however, if the setting considered is too broad, the results can be, from a practical point of view, less than optimal. In this work we study capabilities of an important fundamental resource, namely that of local dissipation.

The most general transformations of quantum states are linear completely positive trace-preserving maps (CPM) also called quantum channels~\cite{Zyczkowski}. A simpler subset consists of transformations that are solutions of local Markovian master equations of the Lindblad type~\cite{GKS,Lindblad}. Lindblad equations, whose generic properties have been studied in e.g.~\cite{Baumgartner:08,Ticozzi:08}, are increasingly recognized as a useful resource, be it for universal computation~\cite{Verstraete:09}, steady-state manipulation~\cite{Zanardi:14}, as an experimental setting, or as a tool for nonequilibrium physics. Characterizing the set of steady states reachable with given resources is a formidable problem with results being usually limited to special cases. For pure steady states simple conditions determine their stationarity~\cite{Zanardi:98,Yamamoto:05,Kraus:08}, while much less is known about mixed steady states, see though Ref.~\cite{Sauer:13}. Frequently an important constraint is locality of interaction, either because it is least costly to implement, or because it naturally arises in weakly coupled systems. Its role on pure steady states has been studied in Ref.~\cite{Ticozzi:12}, and local translational conservation laws in~\cite{JMP:14}. 

In this work we study local one-qubit Lindblad dissipation. First, we efficiently characterize it, leading to a simple geometrical picture akin to the celebrated tetrahedron geometry of qubit channels~\cite{Ruskai:02}. We then demonstrate its usefulness by studying the set of states reachable by arbitrary two qubit Hamiltonian and one-qubit dissipation. We in particular show that the overlap of the steady state with a maximally entangled state, i.e. the singlet fraction, is upper bounded by $(3+\sqrt{5})/8\approx 0.654$. This result sheds light on the influence of local dissipation on non-local quantum properties. An important consequence is that, provided the system's ground state is sufficiently entangled, local dissipation can not cool the system down to low temperatures. There is a ``temperature gap'' below which one can not go. For the particular setting studied this improves on the 3rd law of thermodynamics and various generic zero-temperature unattainability results~\cite{Brumer:13,Masanes:14} or, on ground-state cooling limitations in specific situations, e.g.~\cite{GScooling}. We show that not only is the zero temperature unattainable, but also that finite (low) non-zero temperatures can not be reached with local dissipation.

\section{Local dissipation}

The Lindblad equation is~\cite{GKS,Lindblad}
\begin{eqnarray}
\frac{{\rm d}\rho}{{\rm d}t}&=&{\cal L}(\rho)=\ii[\rho,H]+\cLd(\rho),
\label{eq:Lin}
\end{eqnarray}
where $\cLd(\rho)=\sum_j 2L_j \rho L_j^\dagger-\rho L_j^\dagger L_j-L_j^\dagger L_j \rho$ is a superoperator called a dissipator that depends on a set of traceless Lindblad operators $L_j$. The Lindblad equation, that generates a CPM map via $\Lambda_\tau = {\rm e}^{\cL\,\tau}$, can also be written in a non-diagonal form with a dissipator $\cLd^{(\rm nd)}(\rho)=\sum_{j,k} g_{j,k} (2\tilde{L}_j\rho \tilde{L}_k^\dagger-\rho \tilde{L}_j^\dagger \tilde{L}_k-\tilde{L}_j^\dagger \tilde{L}_k \rho)$, where $\tilde{L}_j$ is a complete set of orthogonal traceless operators. Matrix $g_{j,k}$ is called a Gorini-Kossakowski-Sudarshan (GKS) matrix~\cite{GKS} and has to satisfy $g \ge 0$. 

Expanding Hermitian density operator $\rho$ on $n$ qubits in an orthogonal Hermitian traceless operator basis $F_j$, $\rho=\frac{1}{2^n}\mathbbm{1}+\sum_j c_j F_j$, so that $\rho$ is parametrized by real coherence vector $\mathbf{c}$, dissipator $\cLd(\rho)=\sum_j c_j' F_j$ induces an affine map $\mathbf{c}'=\mathbf{M}\mathbf{c}+\mathbf{\tilde{t}}$, with real matrix $\mathbf{M}$ and real vector $\mathbf{\tilde{t}}$. The unitary part $\ii [\rho,H]$ of the Lindblad equation induces map $\mathbf{c'}=\mathbf{N}\mathbf{c}$, with $\mathbf{N}$ being real antisymmetric. We shall discuss generators $\cLd$ instead of induced channels $\Lambda_t$ because of simpler relations and because they form a convex set while Lindblad channels $\Lambda_\tau$ do not, e.g.~\cite{Wolf:08}.

We are interested in one-qubit dissipators, for which $\mathbf{M}$ is always symmetric~\cite{Lendi:87}, and therefore diagonalizable. Therefore, in an appropriate basis $\cLd$ can be written in a canonical diagonal form (basis $\{\sx,\sy,\sz,\1\}$),
\begin{equation}
\cLdd=
\left( \begin{array}{cccc}
-\frac{q_2+q_3}{2} & 0 & 0 & t_1 \\
0 & -\frac{q_1+q_3}{2} & 0 & t_2 \\
0 & 0 & -\frac{q_1+q_2}{2} & t_3 \\
0 & 0 & 0 & 0\\
\end{array} \right).
\label{eq:Ldiag}
\end{equation}
In its canonical form $\cLdd$ a Lindblad dissipator can be parametrized by $6$ real parameters, $q_j$ and $\mathbf{t}=(t_1,t_2,t_3)$. The central question is, for which values of these parameters (\ref{eq:Ldiag}) is the resulting $\cLdd$ of a Lindblad form $\cLd$ (\ref{eq:Lin}), ie., generates a dynamical semigroup? It is important to note -- and this is the main advantage of the parametrization in Eq.~(\ref{eq:Ldiag}) -- that $\cLd$ is {\em linear} in these parameters, while on the other hand it is {\em quadratic} in Lindblad operators $L_j$. Having linear parametrization of $\cLd$ will greatly simplify all optimization problems, including our application later on.

We write the canonical dissipator $\cLdd$ (\ref{eq:Ldiag}) in terms of a GKS matrix (written in the basis $\{\sx,\sy,\sz\}$)
\begin{equation}
g=
\frac{1}{8}\left( \begin{array}{cccc}
q_1 & -\ii t_3 & \ii t_2 \\
\ii t_3 & q_2 & -\ii t_1 \\
-\ii t_2 & \ii t_1 & q_3 \\
\end{array} \right).
\label{eq:g}
\end{equation}
Sufficient and necessary condition for $\cLdd$ to be of a Lindblad form is that $g \ge 0$. If the dissipator is unital, $\mathbf{t}=0$, the condition $g\ge 0$ translates to three simple conditions $q_{1,2,3}\ge 0$. We are now going to show that for non-unital case a single additional condition is necessary.
\begin{theorem}
One-qubit dissipator in the canonical form given by Eq.(\ref{eq:Ldiag}) represents a Lindblad dissipator iff 
\begin{eqnarray}
\label{eq:unit}
q_i &\ge& 0,\\
1 &\ge& \frac{t_1^2}{q_2 q_3}+\frac{t_2^2}{q_1 q_3}+\frac{t_3^2}{q_1 q_2}.
\label{eq:con}
\end{eqnarray}
(If any of the denominators in the last equation is zero, it must be understood that the corresponding numerator $t_j$ must also be zero, and the term is left out on the RHS.)
\label{thm1}
\end{theorem}
\begin{proof}
$g\ge 0$ iff all eigenvalues $\lambda_j$ of $g$ are non-negative. The characteristic polynomial of Eq.~(\ref{eq:g}) is $p(\lambda)\defeq \det{(g-\lambda \1)}=512\lambda^3-64 C \lambda^2+8(A-t^2)\lambda-(Q-B)$, where $A \defeq q_1q_2+q_1q_3+q_2q_3$, $B \defeq q_1 t_1^2+q_2 t_2^2+q_3 t_3^2$, $C \defeq q_1+q_2+q_3$, $Q \defeq q_1q_2q_3$ and $t^2\defeq \mathbf{t}\cdot\mathbf{t}$. Using Descartes' rule of signs we can infer the maximal number of positive and negative roots by the number of sign changes of coefficients in $p(\lambda)$ and $p(-\lambda)$, respectively. Because $g$ is Hermitian we known that all roots of $p(\lambda)$ are real and we can in fact determine the exact number of negative and positive roots. First, for $g\ge 0$ all diagonal matrix elements $q_j$ must be non-negative, $q_j \ge 0$, which also implies that $C>0$ (as $\tr{g}=C/8$, $C=0$ would imply that $\cLdd\equiv 0$). The coefficient in front of $\lambda^3$ is positive and in order to have three non-negative roots one must also have $A \ge t^2$ and $Q\ge B$. $g$ has one zero eigenvalue if $Q=B$, if in addition $A=t^2$, it has two. It is not possible to have $A=t^2$ and $Q \neq B$ because this would imply less than $3$ positive roots. 

$Q -B\ge 0$, together with $q_j \ge 0$, implies that $Q \ge Q-q_2 t_2^2-q_3 t_3^2\ge q_1 t_1^2$, in turn leading to $q_2 q_3 \ge t_1^2$, provided that $q_1 \neq 0$. Two similar inequalities hold for other index combinations. Summing them together means that, if all $q_j>0$, $Q \ge B$ implies $A \ge t^2$. For $Q=0$, when one or two $q_j$ are zero, inequality $Q\ge B$ must be understood correctly in order to implicate $A\ge t^2$: if $q_1=0$, $Q\ge B$ directly implies that $t_{2,3}=0$ and then, dividing inequality by $q_1$, we get $q_2q_3\ge t_1^2$; if $q_{2,3}=0$ all $t_j=0$. These special cases can be accounted for by instead of $Q \ge B$ writing the inequality as $1 \ge \frac{t_1^2}{q_2 q_3}+\frac{t_2^2}{q_1 q_3}+\frac{t_3^2}{q_1 q_2}$ and understanding that, if any of the denominators is zero, the corresponding $t_j$ must also be zero. We have proved that if $g \ge 0$ then conditions (\ref{eq:unit},\ref{eq:con}) hold. In the opposite direction, if $q_j \ge 0$ and $Q \ge B$, we have seen that $A \ge t^2$ as well as $C>0$ and all the eigenvalues of $g$ are non-negative.
\end{proof}
We see that in addition to simple unital conditions (\ref{eq:unit}) only one additional inequality (\ref{eq:con}) is needed to characterize Lindblad dissipators. The situation is similar to the one for quantum channels~\cite{Ruskai:02} where also a single additional condition is required~\cite{JPA:14}.

{\em Important special cases of $\cLdd$.--} Let us have a look at cases when some $q_j$ are zero. There are only two possibilities that also exhaust all possible $\cLdd$ with one Lindblad operator (ie., $g$ of rank $1$):
\begin{enumerate}
\item Exactly one $q_j$ is zero, say $q_1=0$. As we have seen, this implies $t_2=t_3=0$ and $q_2 q_3\ge t_1^2$. One, or if also $q_2q_3=t_1^2$ two, eigenvalues of $g$ are zero. An example of such dissipator would be one with a single Lindblad operator $L=\sx+\ii \sy$.
\item Exactly two $q_j$ are zero, say $q_2=q_3=0$. In this case all $t_j$ must be zero and the dissipator is unital. Two eigenvalues of the GKS matrix are zero and, up to unitary rotations, we have single Lindblad operator $L \propto \sx$. 
\end{enumerate}
If all $q_j>0$ then one can use inequalities like $q_1 q_2 \ge t_3^2$ to show that one can not have $A=t^2$ (while $Q=B$, ie., $g$ of rank $2$, is still possible). The case (ii) is the only possibility of $\cLdd$ with two-times degenerate steady state.

\subsection{Geometry of one-qubit dissipators} 

Let us now compare the set of channels obtained from Lindblad dissipators to the set of all CPMs. Any one-qubit CPM $\Lambda$ can be brought to a ``diagonal'' form~\cite{Zyczkowski},
\begin{equation}
\Lambda=
\left( \begin{array}{cccc}
\mathbf{D} & \mathbf{v} \\
\mathbf{0} & 1 \\
\end{array} \right),\quad \mathbf{D}={\rm diag}(\lambda_1,\lambda_2,\lambda_3).
\label{eq:channel}
\end{equation}
We shall compare the two sets in terms of allowed $\lambda_j$ (Figs.~\ref{fig:unital} and \ref{fig:nonunital}) and in terms of $v_j$ (Appendix~\ref{app:B}). Conditions which $\lambda_j$ and $v_j$ have to satisfy for $\Lambda$ to be a CPM are well known~\cite{Ruskai:02}. In the unital case ($v_j=0$) $\lambda_j$ must lie within a tetrahedron defined by $4$ corners at $\lambda_j=\pm 1$ and $\lambda_1\lambda_2\lambda_3=1$, whereas for non-unital channels an additional inequality has to hold~\cite{JPA:14}. For diagonal Lindblad generator $\cLdd$ in Eq.(\ref{eq:Ldiag}) it is simple to obtain the corresponding Lindblad channel $\Lambda_\tau={\rm e}^{\cLdd\,\tau}$ (without loss of generality we set $\tau=1$) with parameters in Eq.(\ref{eq:channel}) being given by $q_1=\ln{\frac{\lambda_1}{\lambda_2\lambda_3}}$, $t_1=-v_1 \frac{\ln{\lambda_1}}{1-\lambda_1}$, and analogously for other components. For unital Lindblad channels the set of all $\cLdd$ in the space of $\lambda_j$ is shown in Fig.~\ref{fig:unital}.
\begin{figure}[!t]
\centering \includegraphics[width=2.2in]{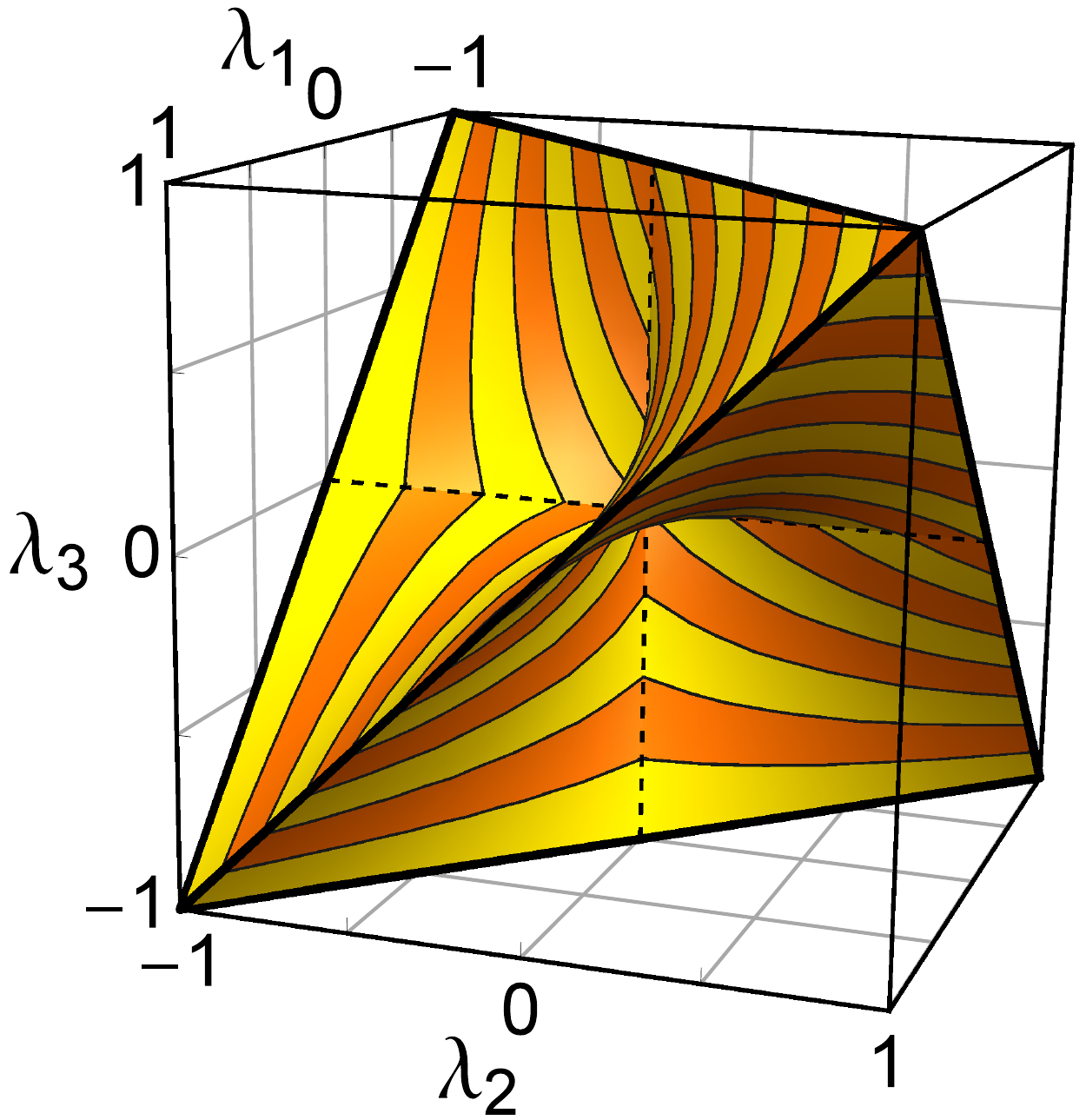}
\includegraphics[width=1.1in]{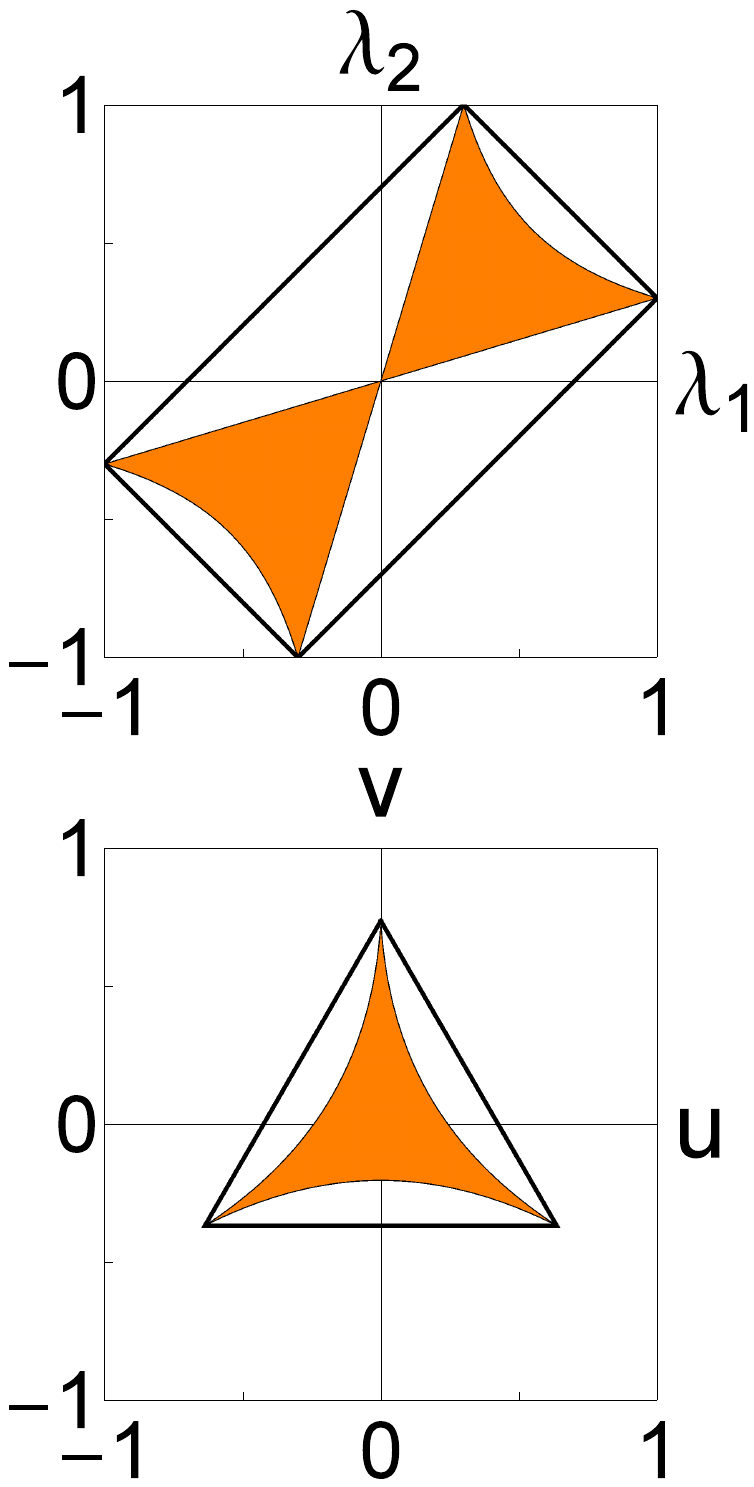}
\caption{(Color online) The set of $\lambda_j$ (\ref{eq:channel}) for one-qubit unital Lindbladian channels $\cLdd$ (colored object). Edges correspond to dissipators $\cLdd$ composed of one Lindblad operator, the surface of two, and the interior from three Lindblad operators. Stripes at the surface, aiding in visualization, are isolines of planes perpendicular to vector $(1,1,1)$. Right top: crosssection at $\lambda_3=0.3$; full rectangle is the CPM' tetrahedron boundary. Right bottom: crosssection perpendicular to $(1,1,1)$ containing point $0.4(1,1,1)$; full triangle is the tetrahedron's edge.}
\label{fig:unital}
\end{figure}
Because $q_j\ge 0$ implies $\lambda_j \ge 0$, in the space of $\lambda_j$ the channel $\Lambda_1$ is always in the first octant. By trivially concatenating such Lindblad channel by a rotation by $\pi$ around one of the three axes one can change the sign of two $\lambda_j$, obtaining symmetrical objects shown in Figs.~\ref{fig:unital} and~\ref{fig:nonunital}. The set of unital Lindblad channels is bounded by hyperbolic paraboloid surfaces like $\lambda_1 \ge \lambda_2 \lambda_3$ due to $q_j\ge 0$, and is of course smaller than the set of all CPMs (tetrahedron). In the 1st octant Lindblad channels fill $\frac{1}{2}$ of the volume of all CPMs, in the whole tetrahedron this fraction is $\frac{3}{8}$. For non-unital channels there are three additional ``shift'' parameters $v_j$. At fixed $v_j$ the set is a pinched rounded tetrahedron, the rounding being essentially determined by Eq.(\ref{eq:con}), pinching by Eq.(\ref{eq:unit}), with logarithms in relations between $q_j$ and $\lambda_j$ adding additional complexity. The set of such non-unital Lindblad channels, together with a ``rounded'' tetrahedron set of non-unital CPMs (see Ref.~\cite{JPA:14} for equations), is shown in Fig.~\ref{fig:nonunital}.
\begin{figure}[!t]
\centering \includegraphics[width=2.2in]{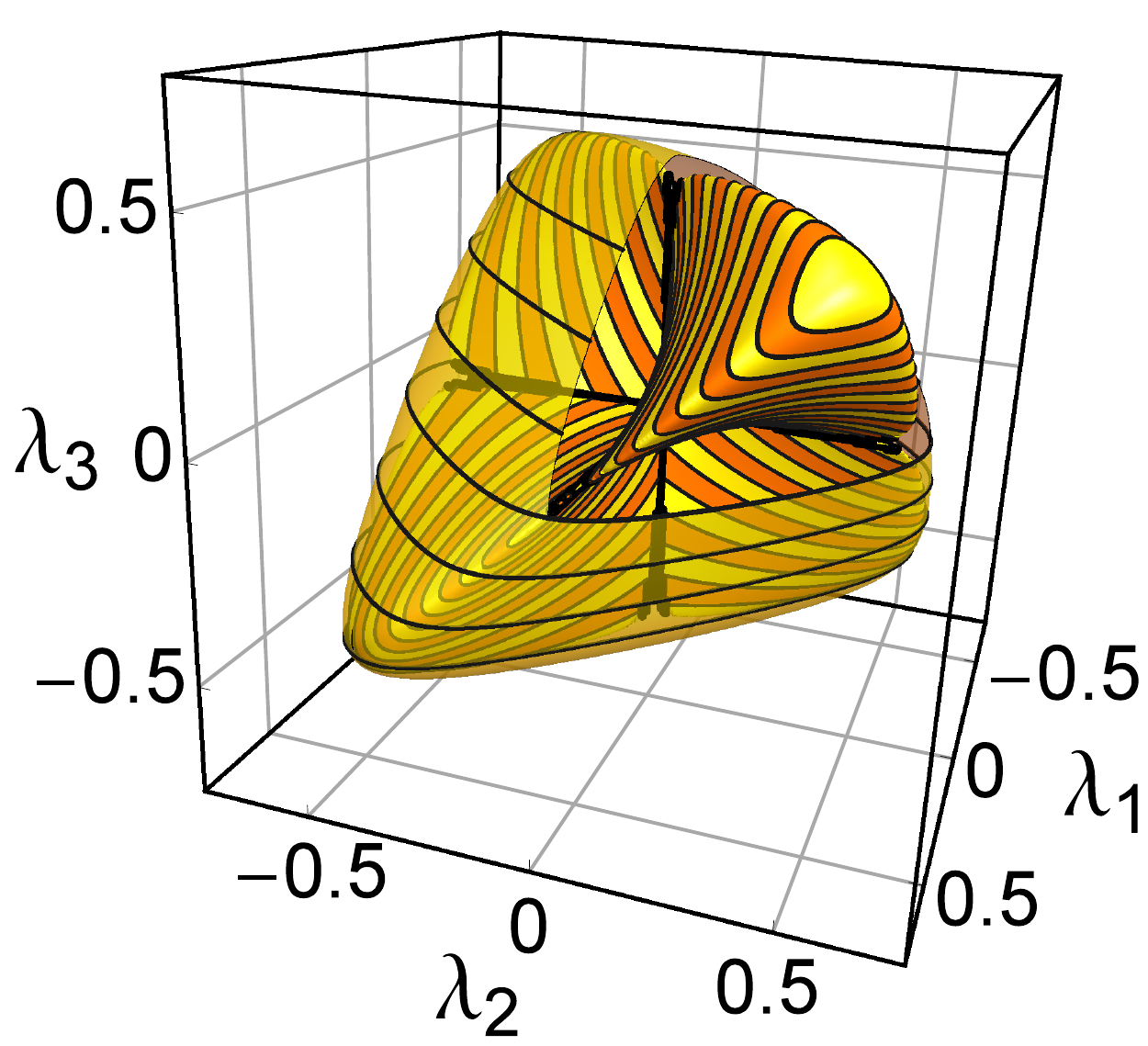}
\includegraphics[width=1.1in]{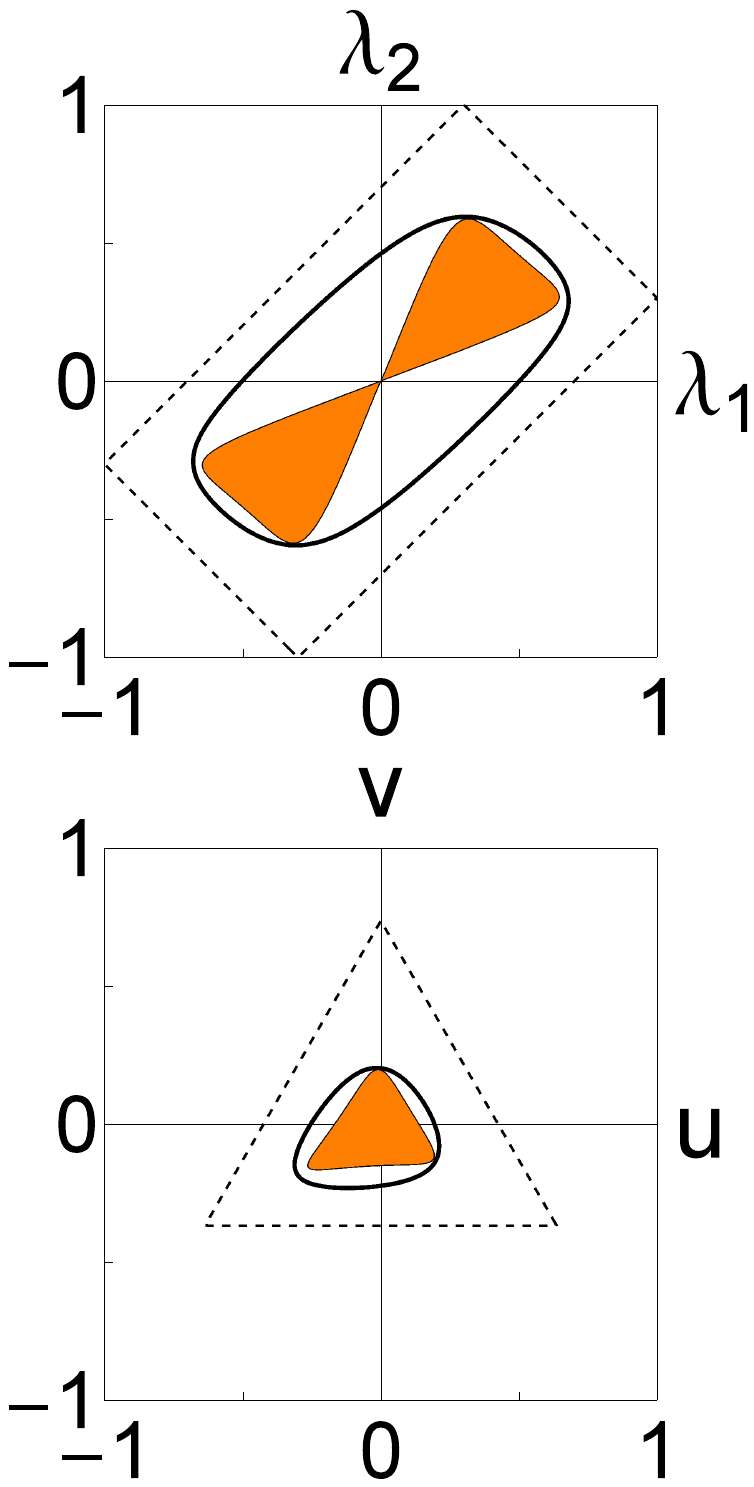}
\caption{(Color online) The sets of one-qubit non-unital Lindblad channels (solid object) and all CPMs (transparent outer shell), both for fixed $\mathbf{v}=\frac{1}{2}(\sin{\theta}\cos{\phi},\sin{\theta}\sin{\phi},\cos{\theta})$ with $\theta=\pi/4$ and $\phi=\pi/3$. The surface corresponds to dissipators made from $2$, the interior from $3$ Lindblad operators. Right: same crossections as in Fig.~\ref{fig:unital}. Dotted lines are tetrahedron edges (unital CPMs), full curves are crossection with the non-unital CPM's surface (transparent shell), colored region is the Lindblad set.}
\label{fig:nonunital}
\end{figure}
One could also plot the two sets for fixed $q_j$ instead of for fixed $v_j$, getting an ellipsoid for the Lindblad case, Eq.(\ref{eq:con}), for details see Appendix~\ref{app:B}. It is interesting to note that various tetrahedron-like geometric objects often pop up when studying different quantum information objects, e.g.~\cite{Zyczkowski,Ruskai:02,tetra}, a common denominator being some form of positivity constraint.

\section{Application}
 
The singlet fraction of a 2-qubit steady state $\rho_\infty$, $\cL(\rho_\infty)=0$, is defined as $F\defeq \langle \psi \vline\, \rho_\infty\, \vline \psi \rangle$, where $\ket{\psi}=(\ket{01}-\ket{10})/\sqrt{2}$. States with high singlet fraction can be used for quantum processing, for instance teleportation~\cite{Horodecki:99}. We want to find the maximal singlet fraction $F$ for a 2-qubit system, maximized over all possible 2-qubit $H$ and 1-qubit dissipators $\cLd$ (this is different than in~\cite{Sauer:13} where the set is studied for a fixed dissipation). The following simple theorem, holding for any number of qubits $n$, will be of help. If $\rho_\infty=\frac{1}{2^n}\1+\sum_j c_j F_j$ is a steady state of Lindblad equation $\cL$ with dissipator $\cLd$ and unitary part given by $H$, acting as $\mathbf{M}\mathbf{c}+\mathbf{\tilde{t}}$, then $\rho'_\infty\defeq \frac{1}{2^n}\1+k\sum_j c_j F_j$ is a steady state of Lindblad equation $\cL'$ having the same $H$ and $\mathbf{M}$ as $\cL$, while the shift in the dissipator $\cLd'$ is $\mathbf{\tilde{t}}'=k\, \mathbf{\tilde{t}}$ (provided, of course, that $\cLd'$ is also a valid Lindblad dissipator). This can be proved by simply verifying that $\cL'(\rho'_\infty)=0$. The length of steady-state coherences, i.e., its purity, can be increased by increasing $\mathbf{\tilde{t}}$. An important consequence is that for fixed $\mathbf{M}$ the largest steady-state overlap with any given state is reached at the maximal possible shift $|\mathbf{\tilde{t}}|$. For 1-qubit $\cLd$ the maximal $F$ is therefore reached either for $\mathbf{\tilde{t}}\equiv 0$ (unital dissipator) or for $g$ of rank $\le 2$.

If $\cLd$ is unital the stationary state coherences $\mathbf{c}$ satisfy $(\mathbf{M}+\mathbf{N})\mathbf{c}=0$. The steady-state subspace is therefore spanned by $\1$ as well as optional pure states corresponding to nontrivial solutions of $(\mathbf{M}+\mathbf{N})\mathbf{c}=0$. It is well known~\cite{Zanardi:98,Yamamoto:05,Kraus:08} that pure state $\ket{\chi}$ can be a steady state of Lindblad equation iff it is an eigenvector of all $L_j$ as well as of $H-\ii \sum_j L_j^\dagger L_j$. Using Schmidt decomposition $\ket{\chi}=\sum_k \mu_k \ket{\alpha_k}_{\rm A}\ket{\beta_k}_{\rm B}$ we see that $L_j \ket{\alpha_k}_{\rm A}=\delta_j \ket{\alpha_k}_{\rm A}$ must hold for all $k$. If the rank of $\ket{\chi}$ is maximal (for a given bipartition such that $L_j$ acts on $A$) $L_j$ must be an identity operator and therefore $\cLd$ is zero. In particular, for a 1-qubit dissipator pure steady state is always separable and therefore $F$ can be at most $1/2$. This maximum can be reached only if there is a single Lindblad operator (two different Pauli matrices do not have a common eigenvector), e.g. $L \sim \sx$, and therefore two $q_j$ are zero. On the other hand, for nonzero shift length $t \defeq | \mathbf{t}|$ at most one $q_j$ can be zero. As a consequence, in the limit of small nonzero shift $t \to \epsilon$ the optimal singlet fraction can be shown to be $1/4$. There is a discontinuous transition in the maximal $F$ from $1/2$ to $1/4$ as one smoothly moves from unital ($t=0$) to non-unital dissipators ($t\neq 0$).   

The singlet fraction $F$ is invariant to any local rotation $U\otimes U$ ($U$ is a 1-qubit unitary) and we can always rotate dissipator to a basis in which the shift vector has only one non-zero component, say $t_1 \neq 0, t_{2,3}=0$. Due to symmetry reasons, in the optimal case, we expect $\cLdd$ in this basis to have $q_2=q_3$ and possibly different $q_1$. Because $t$ also has to be maximal we have in addition $q_2 q_3=t_1^2=1$ (because $H$ is arbitrary we can set $t=1$ without loss of generality). One can argue (Appendix~\ref{app:A}) that for such $\cLdd$ (acting on the 1st qubit) the optimal steady state will be of form $\rho_\infty=\frac{1}{4}\1_1\1_2+c_1 \sx_1\sx_2+c_4 \1_1\sx_2+c_6 \sy_1\sy_2+c_{11} \sz_1\sz_2+c_{13} \sx_1\1_2$. One could try to maximize $F$ subject to necessary conditions~\cite{Sauer:13} $\tr{\rho^r_\infty \cLdd(\rho_\infty)}=0$, $r=1,2,3$, we found it though more convenient (Appendix~\ref{app:A}) to use an alternative approach. Operator equation $\cL(\rho_\infty)=0$ represents a set of $15$ equations that are linear in parameters $d_j$ of $H=\sum_j d_j F_j$ as well as in $c_j$. They can be written as a matrix equation $\mathbf{G}\mathbf{d}=\mathbf{f}$, where the matrix $\mathbf{G}$ as well as the inhomogeneous part $\mathbf{f}$ depend on $\mathbf{c}$ (but not on $\mathbf{d}$). The equation has a solution for $\mathbf{d}$ only if $\mathbf{f}$ is orthogonal to the kernel of $\mathbf{G}$. The kernel of $\mathbf{G}$ can be explicitly calculated, resulting in three constraints (Appendix~\ref{app:A}). The fidelity $F=\frac{1}{4}+c_1+c_6+c_{11}$, subject to these constraints, can be analytically maximized. For fixed $q_{2,3}$ the optimum is always at $q_1=0$, i.e., for single Lindblad operator. The absolute maximum is at $q_2=q_3=1$ when $F_{\rm max}=(1+\varphi)/4\approx 0.6545$, where $\varphi=(1+\sqrt{5})/2$ is the golden mean. Steady state $\rho_\infty$ in this optimal case is of rank $4$ and is entangled. Such optimal $\rho_\infty$ can not be obtained exactly at $q_2=1$, but only in the limit of $q_2 \to 1$, which, however, poses no serious obstacles, see Appendix~\ref{app:A}. As we see, local non-unital dissipation can create mixed entangled steady-states, even though it can not create pure entangled steady-states. Without the locality constraint considered in this work it is, of course, always possible to mitigate entanglement destruction by specific dissipation, like for instance when a common environment is coupled to non-interacting systems~\cite{Braun:02}, or for high-temperature entanglement in driven~\cite{Galve:10} or steady states~\cite{NESS}. 

\section{Consequences} 

The fact that maximal singlet fraction $F_{\rm max}$ is bounded away from $1$ has some important implications. For Hamiltonians whose non-degenerate ground-state has an overlap with a maximally entangled state larger than $F_{\rm max}$, using only local dissipation one can not ``cool'' the systems down to its ground state. For such $H$ there is a fundamental limit to local quantum cooling -- a minimal temperature below which one can not cool. Note that in classical setting, using local cooling by e.g. Langevin bath, there are no obstacles to the lowest attainable energy. In the setting studied this improves on various unattainability results for certain specific situations~\cite{GScooling}. Another interesting consequence is that the steady-state $F$ can put a constraint on a possible dissipation, even if nothing is known about $H$ or $\cL$, namely, if $F$ is larger than $F_{\rm max}$, we immediately know that dissipation can not be local. Also, the maximal $F$ is very sensitive to unitality. For strictly unital local dissipation it is $1/2$, whereas for infinitesimally weak violation of unitality it drops to $1/4$. We expect similar results about lowest attainable energy to hold also for more than 2-qubit systems provided dissipation acts locally on only part of a system. The reason is that local dissipation has only a limited influence on non-local quantum correlations. As a simple example, for an $n$-qubit Heisenberg chain and 1-qubit dissipation one can show that a non-degenerate steady state is always separable and therefore there is again a temperature ``gap''.

\section{Conclusion} 

We have efficiently geometrically characterized single-qubit Lindblad dissipators, opening the door for various optimizations involving local dissipation. To demonstrate its applicability we have calculated the maximal singlet fraction achievable by one-qubit dissipation, showing that it is less than $1$, in turn implying a fundamental limit to local quantum cooling.

I would like to thank B.~\v Zunkovi\v c and H.~C.~F.~Lemos for discussions and acknowledge grant P1-0044 from the Slovenian Research Agency.

\appendix
\section{Maximal singlet fraction for GKS matrix with $t_2=t_3=0$}
\label{app:A}

We want to find out the maximal singlet fraction $F$ for $\cLdd$ with $q_2q_3=1$, $t_1=1$ and $t_2=t_3=0$, for which there are two Lindblad operators, $L_1=(\sqrt{q_2}\sy_1+\ii\,\sz_1/\sqrt{q_2})/\sqrt{8}$ and $L_2=\sqrt{\frac{q_1}{8}}\sx_1$. A necessary condition for $\rho$ to be a steady state of Lindblad equation is that $C_r \defeq \tr{\rho^r_\infty \cLdd(\rho_\infty)}$ are zero for $r=1,2,3$~\cite{Sauer:13}. Let us parametrize $\rho$ as $\rho_\infty=\frac{1}{4}\1_1\1_2-\sum_{j,k=0}^3 c_m \sigma^j_1 \sigma^k_2$, where $m=j+4k+1$ and $\sigma^{j=0,\ldots 3}$ denotes $\{ \sx,\sy,\sz,\1\}$, respectively. For such parametrization $F$ is
\begin{equation}
F=\frac{1}{4}+c_1+c_6+c_{11},
\end{equation}
i.e., it is given by a sum of coefficients in front of $\sx_1\sx_2, \sy_1\sy_2$, and $\sz_1\sz_2$. Conditions $C_r=0$ result in equations that are of $(r+1)$-th order in unknown coefficients $c_m$. One could use the method of Langrange multipliers to solve this constrained optimization problem. However, besides practical solvability issues, there is a fundamental difficulty that the domain of allowed $c$'s is not bounded with $C_{1,2,3}=0$ (the reason is that $\cLdd$ acts only on a single qubit). As a consequence, for instance, solving the resulting Euler-Lagrange equations for optimization of $F$ subject to only $C_1=0$, gives a solution for which $F=\frac{1}{2}$, which, as we shall see, is not the correct maximum. To make the domain bounded one could add an additional constraint, for instance $I\defeq \tr{\rho_\infty^2} \le 1$. One difficulty with using only $C_r=0$ though would still remain. $C_r=0$ is necessary and sufficient if the steady state $\rho$ has non-degenerate spectrum, but only necessary otherwise. As it turns out, the optimal $\rho_\infty$ has in our case a degenerate spectrum (one eigenvalue is twice degenerate). We therefore use a slightly different approach, where though conditions $C_r=0$ will still be used to first infer that a number of coefficients $c_m$ are zero in the optimal case.

Maximizing $F$ subject to $C_1=0$ as well as $I=(\frac{1}{4}+4\sum_m c_m^2) \le 1$ we have a quadratic maximization problem that can be solved exactly. First, one can observe that in $C_1$ some coefficients come only in perfect squares, e.g., $2(q_3+q_1)(c_2^2+c_{10}^2+c_{14}^2+c_{6}^2)+2(q_2+q_1)(c_3^2+c_7^2+c_{15}^2+c_{11}^2)$. Therefore, if we have a solution with nonzero $c_{2,10,14,3,7,15}$ it is always better, meaning we will have higher $F$, to set them to zero and instead increase $c_{6,11}$. In the maximum we will always have $c_{2,3,7,10,14,15}=0$. Then, using Lagrange multipliers one can also show that $c_{5,8,9,12}$ must as well be zero: one has a homogeneous set of linear equations for these coefficients with the only solution being $c_{5,8,9,12}=0$, unless a Langrange multiplier takes a special value in which case though equations for $c_{1,4}$ do not have a solution. Therefore, in the optimal situation only five $c_{1,4,6,11,13}$ are nonzero. The fact that only $c_{1,4,6,11,13}$ appear in the steady state with optimal singlet fraction is actually not very surprising. $F$ depends on $c_{1,6,11}$ so these coefficients will likely appear in the optimal $\rho$. Then, dissipator $\cLdd$ couples $\1_1\1_2$ to $\sx_1\1_2$ as well as $\1_1\sx_2$ to $\sx_1\sx_2$. Coefficients $c_{1,4,6,11,13}$ therefore represent in a way a ``minimal'' set that can satisfy all constraints. Incorporating constraints $C_2$ and $C_3$ into analytical argument is harder so we rather show results of numerical optimization. In Fig.~\ref{fig:normA} we show the dependence of optimal $F$ on the norm $\| c_{\rm A}\|^2=4\sum_{m \in \{2,3,5,7,8,9,10,12,14,15\}} c_m^2$. We can see that the maximum is reached for $\| c_{\rm A}\|=0$, as it was already the case for the single constraint $C_1=0$. Adding constraints $C_2=0$ and $C_3=0$, as well as having $q_2 \neq q_3$ and $q_1 \neq 0$, therefore does not change the conclusion that in the optimum only $c_{1,4,6,11,13}$ are nonzero. Observe also from Fig.~\ref{fig:normA} that adding condition $C_3=0$ to $C_{1,2}=0$ adds very little, for instance, we can analytically show that for $\| c_{\rm A}\|=0$ one gets $F_{\rm opt}\approx 0.65496$ when $C_{1,2}=0, I\le 1$, while $F_{\rm opt} \approx 0.65451$ when $C_{1,2,3}=0, I\le 1$. 
\begin{figure}[!ht]
\includegraphics[width=3in]{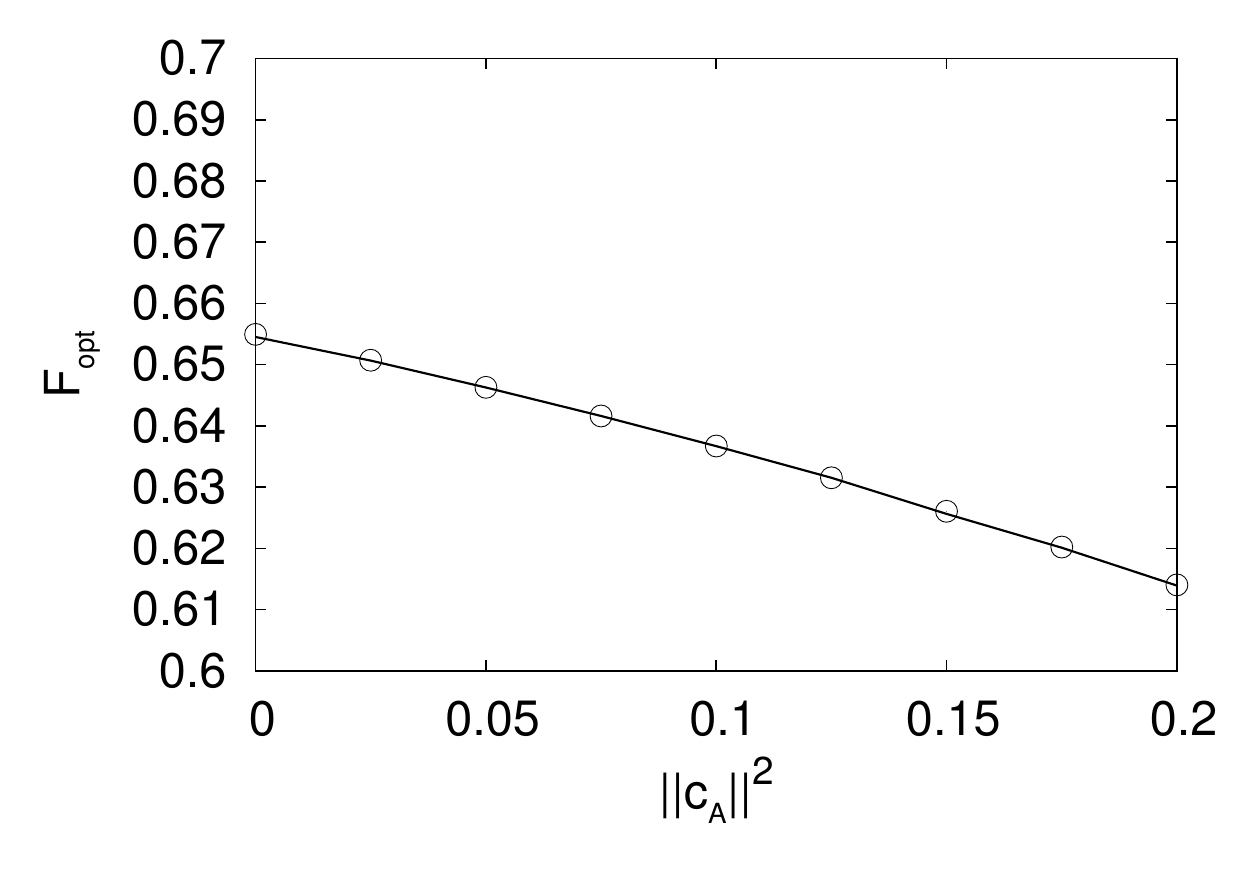}
\caption{Optimal singlet fraction $F$ for different fixed norm $\|c_{\rm A}\|^2$. Circles show the case with constraints $C_1=0, C_2=0, I\le 1$, and the full curve (almost overlapping with circles) shows the case for $C_1=0, C_2=0, C_3=0, I\le 1$; $q_2=q_3=1$, $q_1=0$.}
\label{fig:normA}
\end{figure}

We are therefore left with $5$ nonzero coefficients $c_{1,4,6,11,13}$. Stationary Lindblad equation $\cL(\rho_\infty)=0$ can be written as a matrix equation $\mathbf{G}(\mathbf{c})\mathbf{d}=\mathbf{f}(\mathbf{c})$, for unknown parameters $\mathbf{d}$ of the Hamitonian $H=\sum_j d_j F_j$. The equation has a solution provided $\mathbf{f}(\mathbf{c})$ is orthogonal to the kernel of $\mathbf{G}(\mathbf{c})$. These give sufficient and necessary conditions on coefficients $\mathbf{c}$ in order that the corresponding $\rho_\infty$ is a steady state. Because we have only $5$ remaining unknown $c$'s $\mathbf{G}$ is sufficiently simple so that its kernel, being in general of size 3, can be analytically calculated. First, for fixed $q_{2,3}$ in optimum one always has $q_1=0$, ie. rank $1$ GKS matrix with one Lindblad operator. The three kernel conditions are in this case $q_{23} c_1-4c_4=0$, $1+2q_{23}c_{13}+D(c_6^2-c_{11}^2)=0$ and $q_2 c_{11}+D(c_4 c_{6}+c_{11}c_{13})=0$, with $D \defeq q_3 c_6/(c_4 c_{11}+c_6 c_{13})$ and $q_{23}\defeq q_2+q_3$. These are now sufficiently simple so that $F$ can be analytically maximized. We also note that kernel conditions are stronger than $C_{1,2,3}=0$. The dependence of $F$ on $q_2$ is shown in Fig.~\ref{fig:Fmax}.
\begin{figure}[!h]
\includegraphics[width=2.9in]{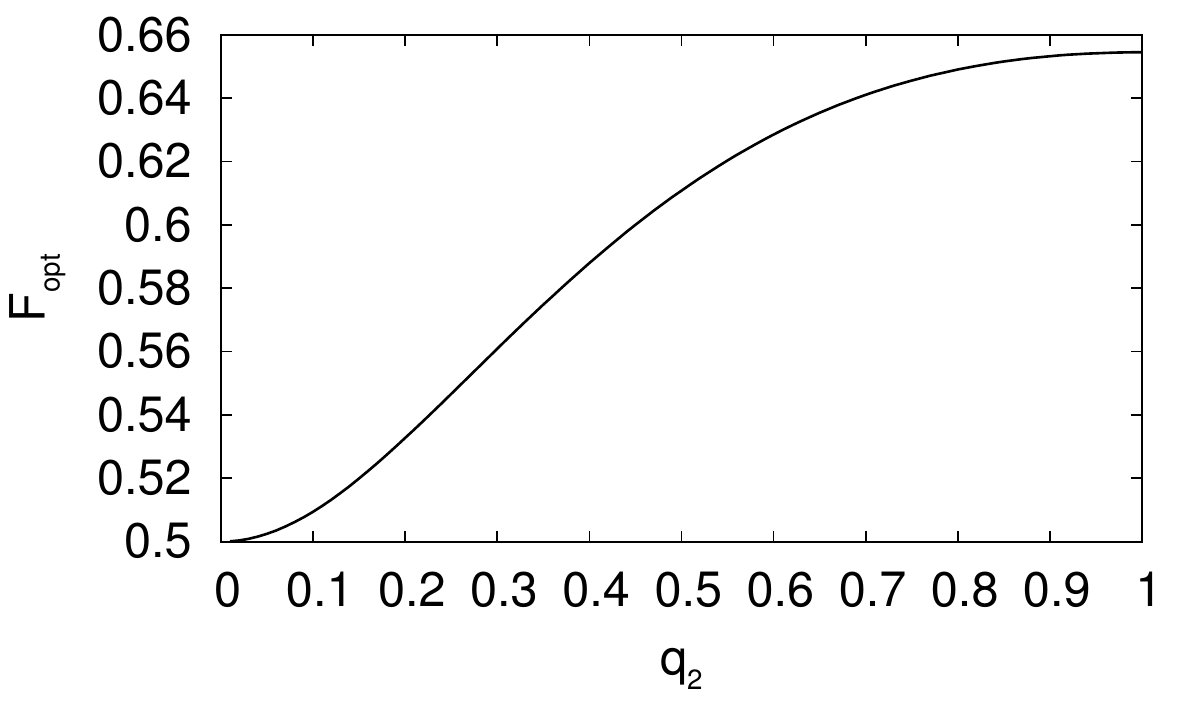}
\caption{The optimal singlet fraction of the steady state for 1-qubit $\cLdd$ with $t_2=t_3=q_1=0$, $t_1=1$, $q_3=1/q_2$.}
\label{fig:Fmax}
\end{figure}
Not very surprisingly, the maximum is achieved for $q_2=q_3=1$. Only at that point is $\cLdd$ invariant to rotations about the $x$-axis, similarly as is the singlet state (for $q_2\neq 1$ the symmetry between $y$ and $z$ is lost). The value of the maximal singlet fraction is $F_{\rm max}=\frac{3+\sqrt{5}}{8}$, reached at $q_2=q_3=1$, while the values of coefficients are $c_1=c_4=-c_{13}=(5+\sqrt{5})/40$, $c_6=c_{11}=1/(4\sqrt{5})$. The optimal point of $q_2=q_3=1$ for which $F=F_{\rm max}$ is in fact degenerate: kernel of $\mathbf{G}$ is in this special case larger and $\mathbf{G}(\mathbf{c})\mathbf{d}=\mathbf{f}(\mathbf{c})$ has no solutions. Optimality can be reached only in the limit $q_2 \to 1$, which, because dependence of $F$ close to $q_2=1$ is quadratic (see Fig.~\ref{fig:Fmax}) poses no serious obstacle. One possibility is to take $c_1=\frac{1}{2}q_2^2(1+2q_2^2/\kappa)/(1+q_2^2)^2, c_4=\frac{1}{4}q_2(1+2q_2/\kappa)/(1+q_2^2), c_6=\frac{1}{4}q_2^2(1+q_2^2)/\kappa, c_{11}=\frac{1}{4}(1+q_2^2)/\kappa, c_{13}=-c_4$, and $H=(\kappa-2q_2^2)/(4(1-q_2^2)(1+q_2^2)^2)[ q_2^2 \sz_1\sy_2+\sy_1\sz_2]$, where $\kappa \defeq \sqrt{1+4q_2^2+10q_2^4+4q_2^6+q_2^8}$, resulting in a singlet fraction $F=q_2^2/(4q_2^2+q_2^4+1-\kappa)$ (which is not the optimal one for $q_2\neq 1$, but approaches the one for $q_2 \to 1$). Taking $q_2=1-\epsilon$, in the limit $\epsilon \to 0$, the expressions simplify to $H=\frac{\sqrt{5}-1}{16\epsilon}(\sz_1\sy_2+\sy_1\sz_2)$, and $c_1=(5+\sqrt{5})/40(1-\epsilon^2/2)$, $c_4=(5+\sqrt{5})/40$, $c_6=(1-\epsilon)/(4\sqrt{5})$, $c_{11}=(1+\epsilon)/(4\sqrt{5})$, $c_{13}=-c_4$, all written to the lowest order in $\epsilon$.

\section{Comparing Lindblad and quantum channels}
\label{app:B}

In the main text we have compared the set of single-qubit Lindblad channels to the set of general qubit channels for the unital case in Fig.~1, and for a fixed shift vector $\mathbf{v}$, Eq.~(6) in main text, in Fig.~2. It is instructive to compare the two sets in a non-unital case also for channel's fixed generalized singular values $\boldsymbol{\lambda}=(\lambda_1,\lambda_2,\lambda_3)$, Eq.~(6).

\begin{figure}[!t]
\includegraphics[width=2.3in]{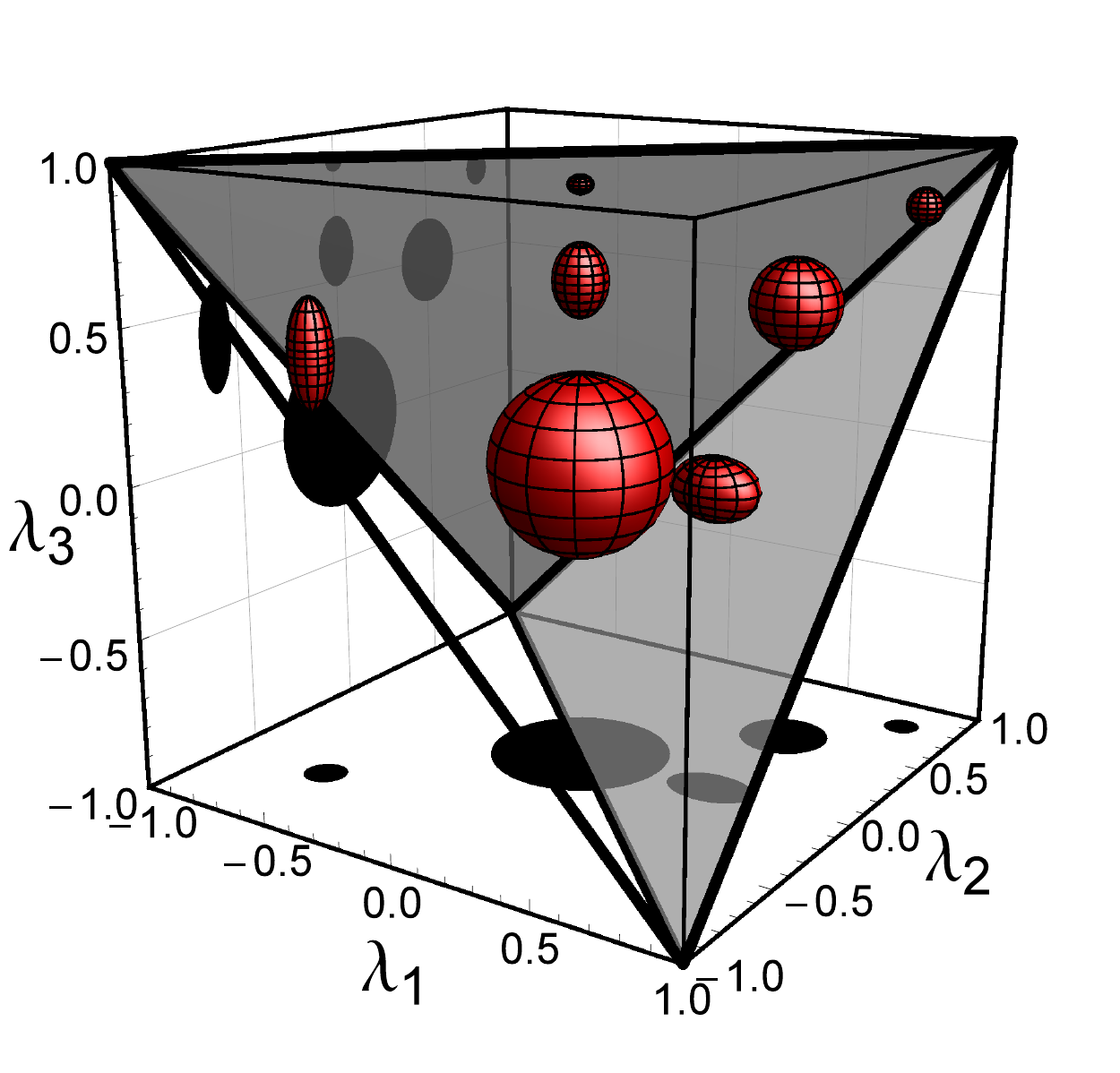}\\
\includegraphics[width=2.3in]{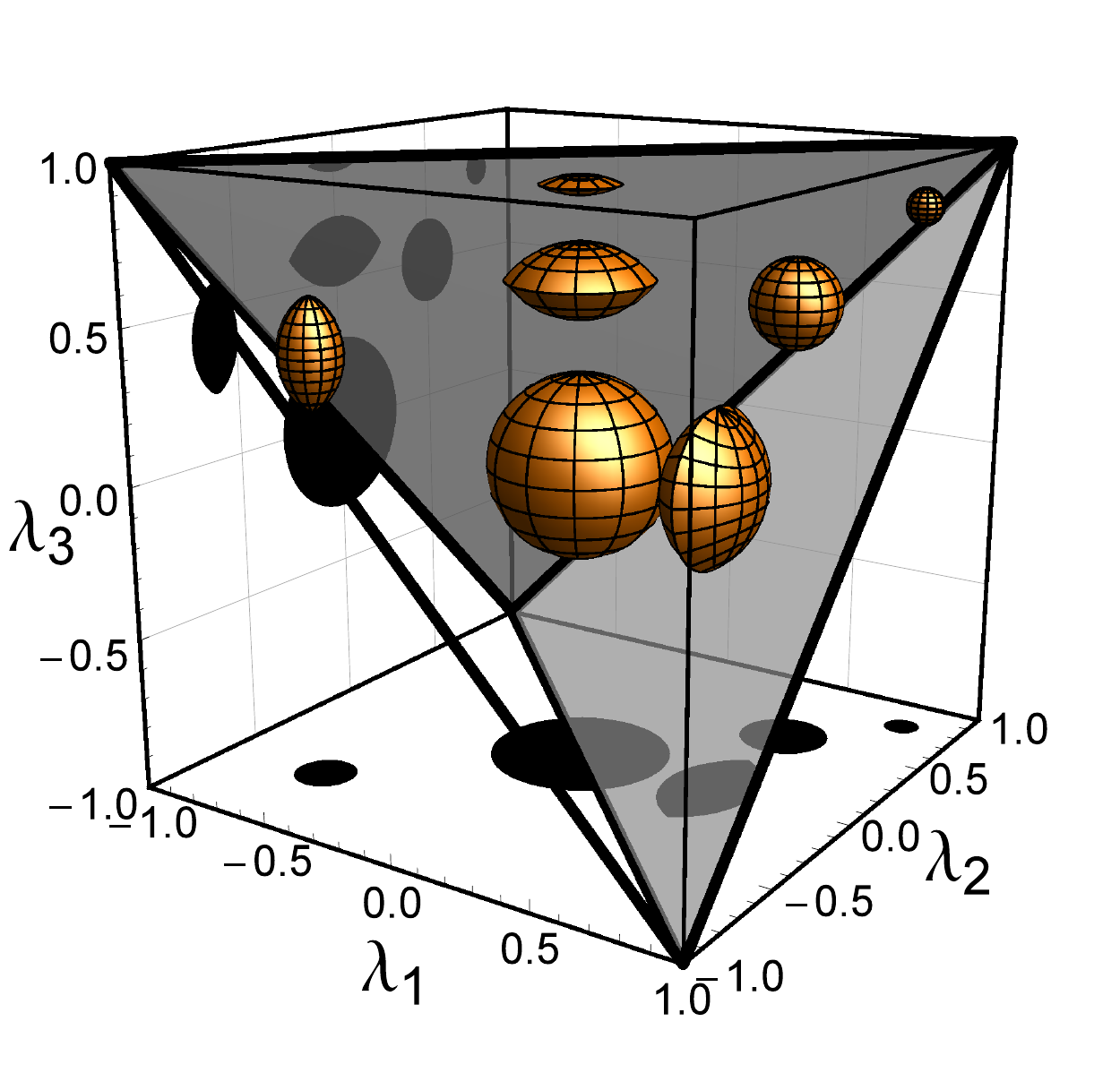}
\caption{(Color online) The sets of allowed channel shifts $\mathbf{v}$ [Eq.~(6)] (scaled by factor $0.3$) for $7$ different diagonal values $\boldsymbol{\lambda}$, namely for $(0,0,0)$, $(0,0,0.6)$, $(0,0,0.9)$, $(0.5,0.5,0.5)$, $(0.8,0.8,0.8)$, $(0.5,0,0)$ and $(-0.6,-0.6,0.4)$. Top: Lindblad channels. Bottom: CPMs. What is shown is a polar plot: maximal $\mathbf{v}$ in a given direction, centered around fixed $\boldsymbol{\lambda}$. In the Lindblad case the surface of ``balls'' corresponds to dissipators with $2$ Lindblad operators, and the interior to $3$. On the diagonal $\boldsymbol{\lambda}\propto (1,1,1)$ the set of allowed $\mathbf{v}$ is the same for CPMs and Lindblad channels (perfect spheres), everywhere else the Lindblad set is smaller.}
\label{fig:balls}
\end{figure}
These sets are shown in Fig.~\ref{fig:balls}. What is shown is a polar plot of the maximal allowed length of the shift vector $\mathbf{v}$, such that the resulting map is completely positive or can be generated by Lindblad evolution, Eq.~(5). Maximal allowed length of $\mathbf{v}$ in a given direction is therefore equal (upto a scale) to the distance between the shown surface and its center, i.e., $\boldsymbol{\lambda}$, around which the ``ball'' is plotted. For Lindblad channels, in the space of generator shifts $\mathbf{t}$, this set would be an ellipsoid, Eq.~(5). In the space of channel shifts $\mathbf{v}$ it is a smooth ellipsoid-like shape seen in the top of Fig.~\ref{fig:balls}. For general channels (CPMs) the set, which is larger or equal to the Lindblad one, can be seen in the bottom of Fig.~\ref{fig:balls}. The boundary of this CPM set is determined by Eq.(30) from Ref.~\cite{JPA:14} and can have non-smooth edges.


\begin{thebibliography}{99}

\bibitem{Nielsen} M.~A.~Nielsen and I.~L.~Chuang, {\em Quantum Computation and Quantum Information}, (Cambridge, 2000).

\bibitem{Zyczkowski} I.~Bengtsson and K.~\. Zyczkowski {\em Geometry of Quantum States: An Introduction to Quantum Entanglement}, (Cambridge, 2006).


\bibitem{GKS} V.~Gorini, A.~Kossakowski, and E.~C.~G. Sudarshan, \tit{Completely positive dynamical semigroups of N-level systems} J.~Math.~Phys. {\bf 17}, 821 (1976).

\bibitem{Lindblad} G.~Lindblad, \tit{On the generators of quantum dynamical semigroups} Commun.~Math.~Phys. {\bf 48}, 119 (1976).

\bibitem{Baumgartner:08} B.~Baumgartner, H.~Narnhofer, and W.~Thirring, \tit{Analysis of quantum semigroups with GKS--Lindblad generators: I. Simple generators} J.~Phys.~A {\bf 41}, 065201 (2008);
B.~Baumgartner and H.~Narnhofer, \tit{Analysis of quantum semigroups with GKS--Lindblad generators: II. General} J.~Phys.~A {\bf 41}, 395303 (2008).

\bibitem{Ticozzi:08} F.~Ticozzi and L.~Viola, \tit{Analysis and synthesis of attractive quantum Markovian dynamics} Automatica {\bf 45}, 2002 (2009).

\bibitem{Verstraete:09} F.~Verstraete, M.~M.~Wolf, and J.~I.~Cirac, \tit{Quantum computation and quantum-state engineering driven by dissipation} Nature Phys. {\bf 5}, 633 (2009).

\bibitem{Zanardi:14} P.~Zanardi and L.~C.~Venuti, \tit{Coherent quantum dynamics in steady-state manifolds of strongly dissipative systems} Phys.~Rev.~Lett. {\bf 113}, 240406 (2014).

\bibitem{Zanardi:98} P.~Zanardi, \tit{Dissipation and decoherence in a quantum register} Phys.~Rev.~A {\bf 57}, 3276 (1998).

\bibitem{Yamamoto:05} N.~Yamamoto, \tit{Parametrization of the feedback Hamiltonian realizing a pure steady state} Phys.~Rev.~A {\bf 72}, 024104 (2005).

\bibitem{Kraus:08} B.~Kraus, H.~P.~B\" uchler, S.~Diehl, A.~Kantian, A.~Micheli, and P.~Zoller, \tit{Preparation of entangled states by quantum Markov processes} Phys.~Rev.~A {\bf 78}, 042307 (2008).

\bibitem{Sauer:13} S.~Sauer, C.~Gneiting, and A.~Buchleitner, \tit{Optimal coherent control to counteract dissipation} Phys.~Rev.~Lett. {\bf 111}, 030405 (2013); S.~Sauer, C.~Gneiting, and A.~Buchleitner, \tit{Stabilizing entanglement in the presence of local decay processes} Phys.~Rev.~A {\bf 89}, 022327 (2014).

\bibitem{Ticozzi:12} F.~Ticozzi and L.~Viola, \tit{Stabilizing entangled states with quasi-local quantum dynamical semigroups}  Phil.~Trans.~R.~Soc. A {\bf 370}, 5259 (2012); F.~Ticozzi and L.~Viola, \tit{Steady-state entanglement by engineered quasy-local Markovian dissipation} Quantum Information and Computation {\bf 14}, 0265 (2014).

\bibitem{JMP:14} M.~ \v Znidari\v c, G.~Benenti, and G.~Casati, \tit{Translationally invariant conservation laws of local Lindblad equations} J.~Math.~Phys. {\bf 55}, 021903 (2014).

\bibitem{Ruskai:02} M.~B.~Ruskai, S.~Szarek, and E.~Werner, \tit{An analysis of completely-positive trace-preserving maps on 2x2 matrices} Lin.~Alg.~Appl. {\bf 347}, 159 (2002).

\bibitem{Brumer:13} L.-A.~Wu, D.~Segal, and P.~Brumer, \tit{No-go theorem for ground state cooling given initial system-thermal bath factorization} Sci.~Rep. {\bf 3}, 1824 (2013);
F.~Ticozzi and L.~Viola, \tit{Quantum resources for purification and cooling: Fundamental limits and opportunities} Sci.~Rep. {\bf 4}, 5192 (2014);

\bibitem{Masanes:14} L.~Masanes and J.~Oppenheim, \tit{A derivation (and quantification) of the third law of thermodynamics} {\tt arXiv:1412.3828} (2014).

\bibitem{GScooling} T.~Feldmann and R.~Kosloff, \tit{Minimal temperature of quantum refrigerators} EPL {\bf 89}, 20004 (2010); 
X.~Wang, S.~Vinjanampathy, F.~W.~Strauch, and K.~Jacobs, \tit{Absolute dynamical limit to cooling weakly coupled quantum systems} Phys.~Rev.~Lett. {\bf 110}, 157207 (2013); 
G.~Benenti and G.~Strini, \tit{Dynamical Casimir effect and minimal temperature in quantum thermodynamics} Phys.~Rev.~A {\bf 91}, 020502(R) (2015).

\bibitem{Wolf:08} M.~M.~Wolf, J.~Eisert, T.~S.~Cubitt, and J.~I.~Cirac, \tit{Assessing non-Markovian quantum dynamics} Phys.~Rev.~Lett. {\bf 101}, 150402 (2008). 

\bibitem{Lendi:87} K.~Lendi, \tit{Evolution matrix in a coherence vector formulation for quantum Markovian master equations for N-level systems} J.~Phys.~A {\bf 20}, 15 (1987).

\bibitem{JPA:14} D.~Braun, \etal{O.~Giraud, I.~Nechita, C.~Pellegrini, and M.~\v Znidari\v c}, \tit{A universal set of qubit quantum channels} J.~Phys.~A {\bf 47}, 135302 (2014).

\bibitem{tetra} I.~Bengtsson, S.~Weis and K.~\. Zyczkowski, \tit{Geometry of the set of mixed quantum states: An apophatic approach} in Geometric Methods in Physics. XXX Workshop, pp 175-197 (Springer, Basel, 2013); P.~D.~Johnson and L.~Viola, \tit{On state versus channel quantum extension problems: Exact results for $U \otimes U \otimes U$ symmetry} J.~Phys.~A {\bf 48}, 035307 (2015).

\bibitem{Horodecki:99} M.~Horodecki, P.~Horodecki, and R.~Horodecki, \tit{General teleportation channel, singlet fraction, and quasidistillation} Phys.~Rev.~A {\bf 60}, 1888 (1999).

\bibitem{Braun:02} D.~Braun, \tit{Creation of entanglement by interaction with a common heat bath} Phys.~Rev.~Lett. {\bf 89}, 277901 (2002).

\bibitem{Galve:10} F.~Galve, L.~A.~Pach\' on, and D.~Zueco, \tit{Bringing Entanglement to the high temperature limit} Phys.~Rev.~Lett. {\bf 105}, 180501 (2010).

\bibitem{NESS} L.~Quiroga, F.~J.~Rodriguez, M.~E.~Ramirez, and R.~Paris, \tit{Nonequilibrium thermal entanglement} Phys.~Rev.~A {\bf 75}, 032308 (2007); L.-A.~Wu and D.~Segal, \tit{Quantum effects in thermal conduction: Nonequilibrium quantum discord and entanglement} Phys.~Rev.~A {\bf 84}, 012319 (2011); M.~\v Znidari\v c, \tit{Entanglement in stationary nonequilibrium states at high energies} Phys.~Rev.~A {\bf 85}, 012324 (2012); B.~Bellomo, and M.~Antezza, \tit{Creation and protection of entanglement in systems out of thermal equilibrium} New~J.~Phys. {\bf 15}, 113052 (2013); S.~Camalet, \tit{Steady Schr\" odinger cat state of a driven Ising chain} Eur.~Phys.~J.~B {\bf 86}, 176 (2013).

\end{thebibliography}
\end{document}